\documentclass[letterpaper,10pt,conference]{ieeeconf}
\IEEEoverridecommandlockouts 
\usepackage{amsmath}
\usepackage{amssymb}
\usepackage{xspace}
\usepackage{adjustbox}
\usepackage{microtype}

\newcommand{\gt}[1]{\xrightarrow{#1}}

\newcommand{\NP}{\text{NP}\xspace}

\newtheorem{theorem}{Theorem}
\newtheorem{lemma}[theorem]{Lemma}
\newtheorem{corollary}[theorem]{Corollary}
\newtheorem{definition}[theorem]{Definition}

\title{\LARGE \bf Secret Protection in Labeled Petri Nets*}

\author{Stefan Haar, Tom\'{a}\v{s} Masopust, and Jakub Ve\v{c}e\v{r}a%
\thanks{*The research of J. Ve\v{c}e\v{r}a has partially been supported by the Palacky University Olomouc under the IGA grant number PrF 2025 018.}
\thanks{S. Haar is with INRIA France {\tt\small stefan.haar@inria.fr}}%
\thanks{T. Masopust and J. Ve\v{c}e\v{r}a are with the Faculty of Science, Palacky University Olomouc, Olomouc, Czechia {\tt\small tomas.masopust@upol.cz}, {\tt\small jakub.vecera01@upol.cz}}%
}

\begin{document}

\maketitle
\thispagestyle{empty}
\pagestyle{empty}

\begin{abstract}
    We study the secret protection problem (SPP), where the objective is to find a policy of minimal cost ensuring that every execution path from an initial state to a secret state contains a sufficient number of protected events. The problem was originally introduced and studied in the setting of finite automata. In this paper, we extend the framework to labeled Petri nets. We consider two variants of the problem: the Parikh variant, where all occurrences of protected events along an execution path contribute to the security requirement, and the indicator variant, where each protected event is counted only once per execution path. We show that both variants can be solved in exponential space for labeled Petri nets, and that their decision versions are \textsc{ExpSpace}-complete. As a consequence, there is no polynomial-time or polynomial-space algorithm for these problems.
\end{abstract}

\section{Introduction}
    Modern systems and applications process sensitive information ranging from personal data to financial credentials. Preventing unauthorized access to such data is a fundamental security objective. This motivation has led to the formalization of the \emph{Secret Protection Problem} (SPP), which captures the requirement that every execution reaching sensitive parts of a system must satisfy a prescribed number of security checks.
    
    The secret protection problem was introduced by Matsui and Cai~\cite{MatsuiC19} and further investigated by Ma and Cai~\cite{MaC22}. In their model, the system is represented by a finite automaton equipped with a designated set of \emph{secret states}. Each secret state is associated with a non-negative integer that specifies the minimum number of security checks that must be performed before the state may be reached. The event set is partitioned into \emph{protectable} and \emph{unprotectable} events. A protectable event may be assigned a security check with an associated non-negative \emph{cost} and a \emph{clearance level} indicating how many clearance units are granted upon its execution. The objective of SPP is to identify a protecting policy---a set of protected events---of minimal cost that ensures that every execution path leading to a secret state meets its assigned security requirement.

    Initial results showed that, for automata, SPP is solvable in polynomial time under certain restrictive assumptions, such as unique event labels for transitions and a uniform clearance level of one~\cite{MaC22}. These results were later extended by Ma and Cai~\cite{MaC25}, who maintained the unique labeling assumption while allowing more sophisticated features, such as dynamic clearance.

		However, it was subsequently shown that dropping the unique event labeling assumption renders SPP \NP-hard. In fact, the hardness persists even if the cost function, clearance function, and security requirements take only binary values. Furthermore, no sub-exponential-time algorithm exists for SPP unless the Exponential Time Hypothesis fails~\cite{MV2025}.
	
    In this paper, we generalize the secret protection problem from finite automata to labeled Petri nets. While an automaton naturally models a single instance of a user operating in the system, Petri nets can represent multiple concurrent instances of the user. Intuitively, each token in the net corresponds to one instance of the user's position within the system. A user may occupy several such positions simultaneously---for instance, when multiple windows of an information system are open in a browser.
    
    Within our framework, a token residing in a place of the net represents an instance of the user being at, or viewing, the corresponding position in the system. If that position is designated as secret, i.e., the respective place has a non-zero security requirement, then the requirement must be satisfied whenever a token enters the place.
    
    We consider two variants of SPP that differ in the way how protected events contribute to meeting security requirements. In the Parikh variant, every occurrence of a protected event in an execution contributes to clearance. In the indicator variant, each protected event contributes only once, regardless of how many times it appears along the execution. The former scenario reflects, for example, multiple browser windows, where each session operates independently and may require repeated checks, whereas the latter corresponds to multiple tabs within a single browser instance, where authentication or other checks may be temporarily reused without re-entering credentials.
    
    Although the two variants yield different complexity results in the automata setting~\cite{MV2025}, we show that, for labeled Petri nets, both problems can be solved in exponential space, and their decision versions are \textsc{ExpSpace}-complete. As a consequence, there is no polynomial-time or polynomial-space algorithm to solve either variant of SPP for labeled Petri nets.

    Our results remain robust even under strong restrictions. In particular, \textsc{ExpSpace}-hardness persists if 
    \begin{enumerate}
        \item every transition carries a unique label, or 
        \item there is only a single protectable event (transition). 
    \end{enumerate}
    This result stands in contrast to the automata setting, where unique event labeling makes the problem tractable in polynomial time.

    Finally, we show that the complexity does not change even if different protectable events follow different counting schemes: Parikh and indicator semantics can be combined within a single instance of SPP without affecting the overall \textsc{ExpSpace}-completeness of the problem.

\section{Preliminaries}
	We assume that the reader is familiar with the basic concepts and results on labeled Petri nets~\cite{Mur89}.

	The set of \emph{natural numbers} (including zero) is denoted by $\mathbb{N}$.
	Let $S$ be a set. The cardinality of $S$ is denoted by $|S|$, and the power set of $S$ is denoted by $2^S$. 

	An \emph{alphabet} $\Sigma$ is a finite nonempty set of \emph{events}. A \emph{string} over $\Sigma$ is a finite sequence of events. The set of all strings over $\Sigma$ is denoted by $\Sigma^*$, and the empty string is denoted by $\varepsilon$. A \emph{language} $L$ over $\Sigma$ is a subset of $\Sigma^*$. For a string $u\in \Sigma^*$ and an event $a\in \Sigma$, the number of occurrences of $a$ in $u$ is denoted by $|u|_a$.
	
	A decision problem belongs to \textsc{ExpSpace} if it can be solved by a deterministic Turing machine using at most $2^{p(n)}$ space on inputs of size $n$, for some polynomial $p$. The class \textsc{coExpSpace} contains all complements of problems in \textsc{ExpSpace}. It is known that \textsc{ExpSpace} is closed under complement, and therefore \textsc{ExpSpace} $=$ \textsc{coExpSpace}.

\subsection{Labeled Petri nets}
	A \emph{Petri net} is a triple $N = (P, T, \mathcal{F})$, where $P$ and $T$ are finite nonempty disjoint sets of \emph{places} and \emph{transitions}, respectively, and 
	$
	    \mathcal{F}\colon (P \times T) \cup (T \times P)\to \mathbb{N}
	$
	is the \emph{flow function} that specifies the (multiplicity of) arcs from places to transitions and from transitions to places. In the rest of the paper, we use the notation $P=\{p_1,\ldots,p_m\}$ and $T=\{t_1,\ldots,t_n\}$. In particular, we use $m$ to denote the number of places and $n$ to denote the number of transitions.
	
	A \emph{marking} is a map $M\colon P \to \mathbb{N}$ that assigns to each place a number of \emph{tokens}. A \emph{transition} $t$ is \emph{enabled} in $M$, denoted by $M\gt{~t~}$, if $M(p) \ge \mathcal{F}(p, t)$ for every place $p \in P$. The \emph{firing} of an enabled transition $t$ in $M$ leads to the marking $M'$, where 
	$
	    M'(p) = M(p) - \mathcal{F}(p, t) + \mathcal{F}(t,p)
	$
	for every $p \in P$, denoted by $M\gt{~t~}M'$.

	For a finite transition sequence $\sigma\in T^*$, we write $M \gt{~\sigma~}$ to indicate that $\sigma$ is enabled at marking $M$, and $M \gt{~\sigma~} M'$ to indicate that $\sigma$ is enabled at $M$ and leads to marking $M'$. A marking $M'$ is \emph{reachable from} $M$ if there exists a transition sequence $\sigma\in T^*$ such that $M \gt{~\sigma~} M'$. 
	
	A \emph{labeled Petri net} (LPN) is a tuple $\mathcal{N} = (P,T,\mathcal{F},\Sigma, \lambda)$, where $(P,T,\mathcal{F})$ is a Petri net, $\Sigma$ is an alphabet of events or labels, and $\lambda \colon T \to \Sigma$ is a \emph{labeling function} that assigns events of $\Sigma$ to transitions. 
	The function $\lambda$ can be seen as a homomorphism  $\lambda\colon T^* \to \Sigma^*$ satisfying $\lambda(\varepsilon)=\varepsilon$ and $\lambda(uv) = \lambda(u)\lambda(v)$, extending thus to the sequences of transitions.
	Let
	$
	    L(\mathcal N,M) = \{ \lambda(\sigma) \mid \sigma \in T^*,\, M \xrightarrow{~\sigma~} \}
	$
	denote the language of observable sequences from $M$.

\section{Secret Protection Problem}
    Let $\mathcal{N} = (P,T,\mathcal{F},\Sigma,\lambda)$ be a labeled Petri net, and let $M \in \mathbb{N}^{m}$ be an initial marking. 
    A {\em security requirement\/} is a function 
    $
        \ell\colon P \to \mathbb{N}
    $
    assigning to each place $p\in P$ a value $\ell(p)$, which specifies the minimum amount of clearance required before a token may enter $p$. We assume that every place initially marked in $M$ has security requirement~$0$.

    In practice, not every transition can be protected. To reflect this situation, the event set $\Sigma$ is partitioned into \emph{protectable events\/} $\Sigma_p$ and \emph{unprotectable events\/} $\Sigma_{up}$, written $\Sigma = \Sigma_p \uplus \Sigma_{up}$. Protectable events correspond to transitions for which protection is possible, while unprotectable events label transitions that cannot be secured. Each protectable event is further associated with
	\begin{itemize}
		\item a \emph{clearance function} $\gamma \colon \Sigma_p \to \mathbb{N}$, specifying the number of clearance units granted upon executing a transition labeled by that event, and
		
		\item a \emph{cost function} $c \colon \Sigma_p \to \mathbb{R}_{\geq 0}$, assigning the implementation cost of protecting each protectable event.
	\end{itemize}
	
    For every protectable event $a\in \Sigma_p$, let $\chi_a \colon \Sigma_p^* \to \mathbb{N}$ be a function that assigns a natural number to each string over $\Sigma_p$. We denote by \[\chi = \{\chi_a \mid a \in \Sigma_p\}\] the family of these functions. In general, the functions in $\chi$ may be arbitrary, but, in this paper, we focus on two specific instances:
    \begin{enumerate}
        \item the \emph{Parikh function} $\chi_a(w) = |w|_a$, which counts every occurrence of $a$ in $w$, and
        \item the \emph{indicator function} 
        \[
            \chi_a(w) = 
            \begin{cases}
                1 & \text{ if $a$ occurs in $w$,}\\ 
                0 & \text{ otherwise,}
            \end{cases}
        \]
        which counts each protectable event in $w$ only once.
    \end{enumerate}

    A single event may label multiple transitions. Under the Parikh semantics, every occurrence of a protectable event contributes to clearance, whereas under the indicator semantics, each protectable event contributes at most once per execution. This reflects situations where repeated executions of the same event either require separate authentication steps or reuse an existing one.

\begin{figure*}
    \centering
    \adjustbox{valign=c}{\includegraphics[scale=1]{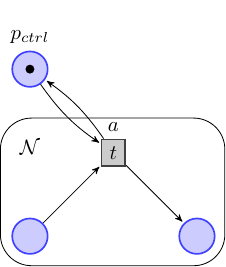}}
    ~~$\leadsto$~~
    \adjustbox{valign=c}{\includegraphics[scale=1]{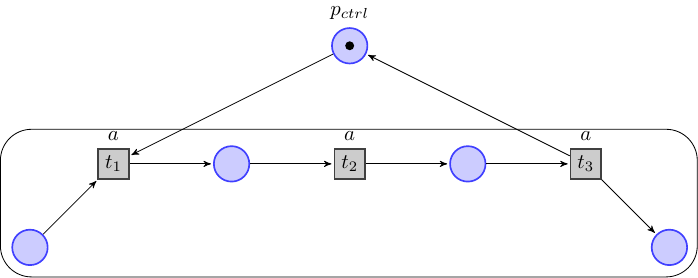}}
    \caption{Sketch of the construction from the proof of Lemma~\ref{lemma2} for $\gamma(a)=3$ (left) and the result of the construction with $\gamma(a)=1$ on the right.}
    \label{fig00}
\end{figure*}

	A \emph{protecting policy\/} (or simply \emph{policy}) is a subset $\mathcal{P} \subseteq \Sigma_p$. 
    A policy $\mathcal{P}$ is \emph{$\chi$-valid\/} if, for every $w \in L(\mathcal{N},M)$ that leads to a marking $M'$ and for every place $p$ with $M'(p) > 0$,
    \[
	    \sum_{a\in \mathcal P} \Bigl(\gamma(a) \cdot \chi_a(w) \Bigr) \ge \ell(p)\,.
	\]
	The cost of a policy $\mathcal{P}$ is the sum of the costs of all protected events, given by 
	\[
		C(\mathcal{P}) = \sum_{a \in \mathcal{P}} c(a)\,.
	\] 
	A policy is \emph{$\chi$-optimal\/} if it is $\chi$-valid and no other $\chi$-valid policy has a lower cost.
	
	We now define the \emph{$\chi$-optimal secret protection problem}.
	\begin{definition}[$\chi$-optimal SPP ($\chi$-SPP)]
		Given a labeled Petri net $\mathcal{N} = (P,T,\mathcal F, \Sigma, \lambda)$, where $\Sigma = \Sigma_p \uplus \Sigma_{up}$, along with an initial marking $M$, a security requirement $\ell$, a clearance function $\gamma$, and a cost function $c$, the objective is to determine the $\chi$-optimal protecting policy $\mathcal P\subseteq \Sigma_p$.
	\end{definition}

	A special case of the secret protection problem, known as the \emph{$\chi$-optimal uniform secret protection problem} ($\chi$-SPP-U), assumes that the clearance level for every event is uniformly set to one, that is, $\gamma(a) = 1$ for every $a \in \Sigma_p$. 
	
	Despite this restriction, $\chi$-SPP-U is equivalent to $\chi$-SPP in the sense that for every instance of $\chi$-SPP, there is an instance of $\chi$-SPP-U such that the costs of the $\chi$-optimal protecting policies coincide (see Lemmata~\ref{lemma2} and~\ref{lemma3}).
		
    Finally, we formulate the decision version of $\chi$-SPP, which we call the budget-constained $\chi$-SPP.
	\begin{definition}[Budget-constrained $\chi$-SPP (BC-$\chi$-SPP)]
	    Given a labeled Petri net $\mathcal{N} = (P,T,\mathcal F, \Sigma, \lambda)$, where $\Sigma = \Sigma_p \uplus \Sigma_{up}$, along with an initial marking $M$, a security requirement $\ell$, a clearance function $\gamma$, a cost function $c$, and a budget limit $W \in \mathbb{N}$, the \emph{budget-constrained $\chi$-optimal secret protection problem} (BC-$\chi$-SPP) asks whether there exists a $\chi$-valid protecting policy $\mathcal{P} \subseteq \Sigma_p$ such that the total cost of the policy satisfies $C(\mathcal{P}) \leq W$.
	\end{definition}

	Analogously, we define the budget-constrained version for $\chi$-SPP-U, denoted by BC-$\chi$-SPP-U.

\section{Parikh-SPP}
    In this section, we analyze the complexity of the optimization and decision versions of $\chi$-SPP and $\chi$-SPP-U under the Parikh semantics; that is, we assume $\chi_a(w) = |w|_a$ for every $a \in \Sigma_p$. We refer to the resulting optimization problems as \emph{Parikh-SPP} and \emph{Parikh-SPP-U}, and to their decision versions as \emph{BC-Parikh-SPP} and \emph{BC-Parikh-SPP-U}.
    
    We first show that Parikh-SPP-U and Parikh-SPP are equivalent.
    \begin{lemma}\label{lemma2}
        For every instance of Parikh-SPP, there is an instance of Parikh-SPP-U such that the costs of the respective Parikh-optimal protecting policies coincide.
    \end{lemma}
	\begin{proof}
	    Intuitively, we simulate the firing of a transition labeled by a protectable event $a$ with clearance $\gamma(a)=k>1$ by a sequence of $k$ transitions, each with uniform clearance~$1$. Consider an instance of Parikh-SPP given by a labeled Petri net $\mathcal N$. The construction of an equivalent instance of Parikh-SPP-U proceeds in two steps.

	    \textbf{Step 1.} Create a copy $\mathcal N'$ of $\mathcal N$ and add a control place $p_{ctrl}$ initially marked with a single token. Connect $p_{ctrl}$ to every transition of $\mathcal N'$ by a self-loop (see Fig.~\ref{fig00}, left). Clearly, this modification preserves the set of firing sequences.
        
	    \textbf{Step 2.} For every transition $t$ of $\mathcal N$ labeled by an event $a$ with $\gamma(a)=k>1$, we replace $t$ in $\mathcal N'$ by a chain of $k$ fresh transitions $t_1,\dots,t_k$ and $k-1$ fresh places $p_{t_1},\dots,p_{t_{k-1}}$. For each place $p$ of $\mathcal N$, we set
        $
            \mathcal F(p,t_1)=\mathcal F(p,t)\quad\text{and}\quad \mathcal F(t_k,p)=\mathcal F(t,p).
        $
        For $i=1,\dots,k-1$, we define
        $
            \mathcal F(t_i,p_{t_i})=\mathcal F(p_{t_i},t_{i+1})=1,
        $
        and add
        $
            \mathcal F(p_{ctrl},t_1)=1 \text{ and } \mathcal F(t_k,p_{ctrl})=1.
        $
        We label each $t_i$ by the event $a$ and set $\gamma(a)=1$; see Fig.~\ref{fig00} for the case $k=3$.

	    By construction, every firing of $t$ in $\mathcal N$ is simulated in $\mathcal N'$ by the sequence $t_1,\dots,t_k$, and vice versa. The control place $p_{ctrl}$ ensures that once $t_1$ fires, no other transition (except $t_2,\dots,t_k$ in order) can fire until $t_k$ restores the token in $p_{ctrl}$; hence the simulation is atomic and introduces no spurious interleavings.
	    
	    Finally, protecting policies and their costs are preserved: a policy $\mathcal P\subseteq\Sigma_p$ for $\mathcal N$ corresponds to the same policy $\mathcal P$ in $\mathcal N'$ (the new transitions carry the same labels). Validity is preserved because each occurrence of $a$ in $\mathcal N$ corresponds to $k$ occurrences in $\mathcal N'$, each contributing one unit under uniform clearance. Thus, an optimal policy of $\mathcal N$ is an optimal policy of $\mathcal N'$, which completes the proof.
    \end{proof} 

    Our next result makes use of the increasing fragment of Yen's path logic~\cite{yen1992unified}, whose satisfiability problem is known to be \textsc{ExpSpace}-complete~\cite{ATIG2011}. For completeness, we briefly recall this fragment.

    Let $\mu_1, \mu_2, \dots$ be variables representing markings, and let $\sigma_1, \sigma_2, \dots$ be variables representing finite sequences of transitions. Every mapping $c \in \mathbb{N}^{m}$ is a \emph{term}, where $m$ denotes the number of places in a Petri net and thus the dimension of the markings. Terms are closed under addition and subtraction: for all $j > i$, if $\mu_i$ and $\mu_j$ are marking variables, then $\mu_j - \mu_i$ is a term, and if $T_1$ and $T_2$ are terms, then $T_1 + T_2$ and $T_1 - T_2$ are also terms.

    Given terms $T_1$ and $T_2$ and places $p_1, p_2 \in P$, expressions of the form
    $
        T_1(p_1) = T_2(p_2),\ T_1(p_1) < T_2(p_2),\ T_1(p_1) > T_2(p_2)
    $
    are called \emph{marking predicates}. More generally, a \emph{predicate} is any positive Boolean combination of marking predicates.
    A {\em path formula\/} is a formula of the form
    \begin{multline*}
        (\exists \mu_1,\dots, \mu_n)
        (\exists \sigma_1, \sigma_2,\dots, \sigma_n) \\
        (\mu_0 \xrightarrow{\sigma_1} \mu_1 \xrightarrow{\sigma_2} \cdots \xrightarrow{\sigma_n} \mu_n) \\
        \land
        \varphi(\mu_1,\ldots,\mu_n,\sigma_1,\ldots,\sigma_n)
    \end{multline*}
    where $\varphi$ is a predicate over the markings and transition sequences that implies $\mu_n \ge \mu_1$. The condition $\mu_n \ge \mu_1$ expresses that the final marking $\mu_n$ dominates the marking $\mu_1$, capturing the essence of the \emph{increasing fragment} in Yen’s path logic.
    
    We now prove that, in the setting of labeled Petri nets, the Parikh-optimal secret protection problem and its uniform variant are both \textsc{ExpSpace}-complete.
    \begin{theorem}\label{Parikh-SPP-complexity-membership}
        Parikh-SPP for labeled Petri nets is solvable in exponential space.
    \end{theorem}
    \begin{proof}
        To prove membership of Parikh-SPP in \textsc{ExpSpace}, let $\mathcal N$ be a labeled Petri net with $m$ places $p_1, \dots, p_m$. Since a protecting policy is a subset of the labels of transitions, we can enumerate all candidate policies in exponential space. For each such policy, we verify Parikh-validity as described below. To compute an optimal solution, we keep track of the Parikh-valid policy of minimum cost. The cost of a given policy can be computed in polynomial time.
        
        Let $\mathcal{P}$ be a fixed protecting policy. From $\mathcal{N}$, we construct a labeled Petri net $\mathcal{N}_{\mathcal{P}}$ by adding a place $p_c$ that serves as a counter for occurrences of events in $\mathcal{P}$ along a firing sequence. Specifically, whenever a transition labeled by some $a \in \mathcal{P}$ fires, $\gamma(a)$ tokens are added to $p_c$; see Fig.~\ref{fig01}.
        
        \begin{figure}
          \centering
          \includegraphics[scale=.99]{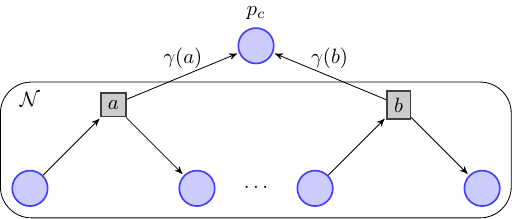}
          \caption{Sketch of the construction from the proof of Theorem~\ref{Parikh-SPP-complexity-membership}.}
          \label{fig01}
        \end{figure}

        To verify Parikh-validity of $\mathcal P$, we express its violation via the increasing fragment of Yen's path logic as follows:
        \begin{multline*}
            \exists \mu_1 \exists \sigma \Bigl( 
                (\mu_0 \xrightarrow{\sigma} \mu_1) 
                \land \\
                \bigvee_{i=1}^{m} \bigl(\mu_1(p_i) > 0 \land (\mu_1-\mu_0)(p_c) < \ell(p_i)\bigr) \Bigr)\,.
        \end{multline*}

        The formula asks whether there exists a firing sequence $\sigma$ from the initial marking $\mu_0$ to a marking $\mu_1$ such that (i) some place $p_i$ of $\mathcal{N}$ contains at least one token in $\mu_1$, and (ii) the number of tokens accumulated in the counter place $p_c$ is strictly below the required clearance~$\ell(p_i)$. Equivalently,
        \[
	        \sum_{a\in \mathcal P} \Bigl(\gamma(a) \cdot \chi_a(\lambda(\sigma)) \Bigr) < \ell(p_i)\,,
        \]
        which captures a violation of Parikh-validity.
        
        Since $\mu_0(p_c)=0$, the term $(\mu_1 - \mu_0)(p_c)$ is used only to conform to the syntax of Yen's framework; it could be replaced by $\mu_1(p_c) < \ell(p_i)$. The quantity $\ell(p_i)$ is admissible because $\ell$ is a constant vector and may be referenced in the logic.
               
        Satisfiability for the increasing fragment of Yen's path logic is \textsc{ExpSpace}-complete~\cite{ATIG2011}. Hence, checking the existence of such a violating run is decidable in exponential space. By closure of \textsc{ExpSpace} under complement, verifying Parikh-validity of $\mathcal{P}$ is also in \textsc{ExpSpace}.
        
        Since each policy can be checked using exponential space, and the total number of candidate policies is exponential, we conclude that Parikh-SPP is solvable in exponential space.
    \end{proof}
    
    As an immediate consequence of Theorem~\ref{Parikh-SPP-complexity-membership}, we have the following corollary.
    \begin{corollary}
        Parikh-SPP-U for labeled Petri nets is solvable in exponential space.
        $\hfill$\QED
    \end{corollary}

    We now show that both BC-Parikh-SPP and BC-Parikh-SPP-U are \textsc{ExpSpace}-complete. Clearly, it is sufficient to prove the result for BC-Parikh-SPP-U.
    \begin{theorem}\label{Parikh-SPP-complexity}
        BC-Parikh-SPP-U for labeled Petri nets is an \textsc{ExpSpace}-complete problem.
    \end{theorem}
    \begin{proof}
        We reduce the coverability problem, which is \textsc{ExpSpace}-complete~\cite{esparza}, to BC-Parikh-SPP-U. Given a Petri net $N$ and a pair of markings $(M_0, M)$, the \emph{coverability} problem asks whether there is a marking $M'$ reachable from $M_0$ such that $M'$ covers $M$, i.e., whether $M_0 \xrightarrow{\sigma} M'$ with $M' \ge M$.

        Let $\mathcal{N}=(P,T,\mathcal F)$ and $(M_0, M)$ be an instance of the coverability problem. We construct a new Petri net $\mathcal{N}'$ as a copy of $\mathcal{N}$, and extend it by adding a new transition $t_{\mathit{new}}$ and a new place $p_{\mathit{new}}$. The new place is initially marked with zero tokens. For every place $p\in P$, we set
        \[
            \mathcal{F}(p, t_{\mathit{new}}) = M(p)
            \quad\text{and}\quad
            \mathcal{F}(t_{\mathit{new}},p_{\mathit{new}}) = 1,
        \]
        with all other entries of the flow function unchanged. The initial marking $M_0'$ of $\mathcal{N}'$ is obtained from $M_0$ by setting $M_0'(p) = M_0(p)$ for all $p \in P$ and $M_0'(p_{\mathit{new}}) = 0$. See Fig.~\ref{pnhardproof} for an illustration.
        
        \begin{figure}
          \centering
          \includegraphics[scale=1]{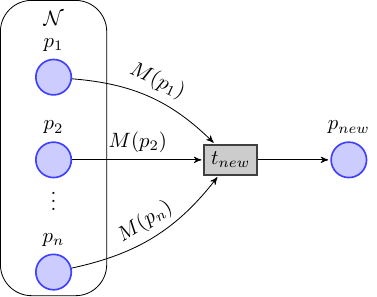}
          \caption{Sketch of the hardness construction from the proof of Theorem~\ref{Parikh-SPP-complexity}.}
          \label{pnhardproof}
        \end{figure}
        
        Next, we make all transitions observable by defining the labeling function $\lambda \colon T \cup \{t_{\mathit{new}}\} \to T \cup \{t_{\mathit{new}}\}$ as the identity, i.e., $\lambda(t) = t$ for every $t \in T \cup \{t_{\mathit{new}}\}$. We designate $t_{\mathit{new}}$ as the only protectable event, i.e., the set of protectable events is $\Sigma_p = \{t_{\mathit{new}}\}$. We set $\gamma(t_{\mathit{new}}) = 1$, $c(t_{\mathit{new}}) = 1$, $\ell(p_{\mathit{new}}) = 2$, and the budget limit to $W = 1$.
        
        Thus, the only possible protecting policies are $\mathcal{P} = \emptyset$ or $\mathcal{P} = \{t_{\mathit{new}}\}$, both of which have cost at most $W = 1$.

        By construction, the transition $t_{\mathit{new}}$ can fire in $\mathcal{N}'$ starting from $M_0'$ if and only if $M$ is coverable from $M_0$ in $\mathcal{N}$. Indeed, $t_{\mathit{new}}$ requires at least $M(p)$ tokens in each original place $p \in P$, and produces one token in $p_{\mathit{new}}$.        

        If $t_{\mathit{new}}$ fires, then $p_{\mathit{new}}$ contains one token. However,
        \[
            \ell(p_{\mathit{new}}) = 2 > 1 = \gamma(t_{\mathit{new}})\,,
        \]
        and hence every marking that enables $t_{\mathit{new}}$ violates the security requirement, regardless of whether $t_{\mathit{new}}$ is in $\mathcal{P}$. This means that the resulting instance is \emph{not} Parikh-valid.

        Conversely, if $t_{\mathit{new}}$ is never enabled, then $p_{\mathit{new}}$ remains unmarked, and the security requirement is vacuously satisfied for every policy of cost at most $W$.
        
        Therefore, the marking $M$ is coverable in $\mathcal{N}$ from $M_0$ if and only if the instance of BC-Parikh-SPP\nobreakdash-U defined above does \emph{not} satisfy Parikh-validity. Hence BC-Parikh-SPP\nobreakdash-U is \textsc{coExpSpace}-hard. Since \textsc{ExpSpace} is closed under complement and BC-Parikh-SPP-U belongs to \textsc{ExpSpace} by Theorem~\ref{Parikh-SPP-complexity-membership}, the proof is complete.
    \end{proof}

\section{Indicator-SPP}
    In this section, we analyze the complexity of the optimization and decision variants of $\chi$-SPP and $\chi$-SPP-U under the indicator semantics. Specifically, for every $a \in \Sigma_p$, we assume
    \[
        \chi_a(w) =
            \begin{cases}
                1 & \text{if $a$ occurs in $w$} \\
                0 & \text{otherwise.}
            \end{cases}
    \]
    We refer to the corresponding optimization problems as \emph{Indicator-SPP} and \emph{Indicator-SPP-U}, and to their decision versions as \emph{BC-indicator-SPP} and \emph{BC-indicator-SPP-U}.

    We first observe that the uniform and non-uniform variants are equivalent.
    \begin{lemma}\label{lemma3}
        For every instance of Indicator-SPP, there is an instance of Indicator-SPP-U such that the costs of the respective indicator-optimal protecting policies coincide.
    \end{lemma}
    \begin{proof}
        The construction follows the same principle as in the proof of Lemma~\ref{lemma2}, with one modification. Suppose a transition $t$ in the original net is labeled by an event $a \in \Sigma_p$ with $\gamma(a) = k > 1$ and cost $c(a)$. In the uniform setting, we replace $t$ with a sequence of $k$ transitions $t_1,\dots,t_k$ and introduce $k$ fresh labels $a_1,\dots,a_k$. We assign
        $
            \gamma(a_i) = 1
        $
        for all $i = 1,\dots,k$, and set the costs $c(a_1) = c(a)$, and $c(a_2) = \dots = c(a_k) = 0$. Thus, under indicator semantics, the original event $a$ contributes clearance $\gamma(a) = k$ if and only if all events $a_1,\dots,a_k$ occur along the corresponding execution, each contributing a single unit. Since exactly one of the newly introduced events carries the original cost, the total cost of an optimal protection policy is preserved.
        It follows that every optimal solution to the original instance of Indicator-SPP has the same cost as the corresponding optimal solution to the constructed instance of Indicator\nobreakdash-SPP\nobreakdash-U.    
    \end{proof}

    We now show that Indicator-SPP, and hence also Indicator-SPP-U, can be solved in exponential space.
    \begin{theorem}\label{Indikator-SPP-complexity}
        Indicator-SPP and Indicator-SPP-U for labeled Petri nets are solvable in exponential space.
    \end{theorem}
    \begin{proof}
        We show that Indicator-SPP is decidable in \textsc{ExpSpace}; the same argument applies to Indicator-SPP-U. Let $\mathcal N=(P,T,\mathcal F,\Sigma,\lambda)$ be a labeled Petri net with $P=\{p_1,\dots,p_m\}$. Since a protecting policy is a subset of the protectable events, we can enumerate all candidate policies $\mathcal P\subseteq\Sigma_p$ using exponential space. For each fixed candidate policy $\mathcal P$, we show how to verify indicator-validity in exponential space; the optimal solution can be obtained by keeping the cheapest policy that is declared valid.

        Fix a candidate policy $\mathcal P$. We construct from $\mathcal N$ a labeled Petri net $\mathcal N_{\mathcal P}$ that records, along any run, which protectable events from $\mathcal P$ have occurred at least once, using a single global counter place ${\mathit{counter}}$. The construction is as follows (see Fig.~\ref{fig02} for intuition).

        \begin{figure}
            \centering
            \includegraphics[scale=1]{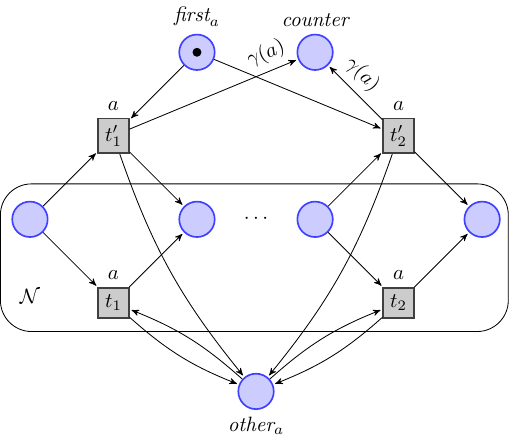}
            \caption{Sketch of the reduction from the proof of Theorem~\ref{Indikator-SPP-complexity}.}
            \label{fig02}
        \end{figure}  

        \begin{enumerate}
            \item Start with a copy of $\mathcal N$.
            
            \item Add one fresh place ${\mathit{counter}}$ (the global counter) initialized with $0$ tokens.
            
            \item For every protectable event $a\in\Sigma_p$, we introduce two fresh places $\mathit{first}_a$ and $\mathit{other}_a$. Initialize $\mathit{first}_a$ with one token and $\mathit{other}_a$ with zero tokens.
            
            \item For every transition $t\in T$ with $\lambda(t)=a\in\Sigma_p$, we add a copy $t'$ that has the same arcs to the original places of $\mathcal N$ (so $t'$ simulates the effect of $t$ on the original marking). Specificaly, for every $p\in P$, we set
            \[
                \mathcal F(p,t')=\mathcal F(p,t) \text{ and } \mathcal F(t',p)=\mathcal F(t,p).
            \]
            Additionally, for the copy $t'$, we add the arcs
            \begin{align*}
                & \mathcal F(\mathit{first}_a,t')=1, \  
                \mathcal F(t',\mathit{other}_a)=1, \text{ and }\\
                & \mathcal F(t',{\mathit{counter}})=\gamma(a).
            \end{align*}
            Thus, $t'$ can fire only while a token is present in $\mathit{first}_a$, and upon firing it moves that token to $\mathit{other}_a$ and deposits $\gamma(a)$ tokens to the global counter.
            
            \item Finally, we add self-loop arcs between each original transition $t$ with label $a$ and the place $\mathit{other}_a$, i.e.,
            \[
                \mathcal F(\mathit{other}_a,t)=\mathcal F(t,\mathit{other}_a)=1,
            \]
            so that once $\mathit{other}_a$ is marked, the original transitions labeled by $a$ remain enabled as in the original net but their copies $t'$ are disabled.
        \end{enumerate}
        
        Intuitively, for every label $a$ the first occurrence of an $a$-labeled transition along a run can be performed by its copy $t'$, which (and only then) contributes $\gamma(a)$ tokens to the single global counter. Subsequent occurrences of label $a$ in the same run are simulated by the original transitions (enabled via $\mathit{other}_a$) but do not increase the counter. Hence, in $\mathcal N_{\mathcal P}$ the global counter ${\mathit{counter}}$ correctly records the total clearance contributed under indicator semantics by the set $\mathcal P$ along any run.
        
        To verify that $\mathcal P$ is \emph{not} indicator-invalid (i.e., violates the security requirement), we search for a run that reaches some marking in which some place $p_i$ is occupied while the counter holds fewer than $\ell(p_i)$ tokens. This condition can be expressed in the increasing fragment of Yen's path logic by the formula
        \begin{multline*}
            \exists \mu_1 \exists \sigma \Bigl((\mu_0 \xrightarrow{\sigma} \mu_1)\ \land\ \\
            \bigvee_{i=1}^{m}\bigl(\mu_1(p_i)>0 \land \mu_1({\mathit{counter}})<\ell(p_i)\bigr)\Bigr).
        \end{multline*}
        By construction, the formula is satisfiable in $\mathcal N_{\mathcal P}$ if and only if there is a run of $\mathcal N$ whose indicator-clearance (w.r.t.\ $\mathcal P$) is insufficient for some place reached by that run; thus satisfiability exactly captures indicator-invalidity of $\mathcal P$.
        
        Satisfiability of formulas in the increasing fragment of Yen's path logic is decidable in \textsc{ExpSpace}~\cite{ATIG2011}. Therefore, for each fixed policy $\mathcal P$, we can determine indicator-validity using exponential space. Combining this with the fact that there are at most exponentially many candidate policies, we obtain an \textsc{ExpSpace} algorithm for Indicator-SPP. The same argument applies to Indicator-SPP-U.
      
        Hence, Indicator-SPP and Indicator-SPP-U are solvable in exponential space.
    \end{proof}

    As an immediate consequence of this result and the proof of Theorem~\ref{Parikh-SPP-complexity}, we obtain the following result.
    \begin{theorem}
        BC-indicator-SPP and BC-indicator-SPP-U for labeled Petri nets are \textsc{ExpSpace}-complete.
    \end{theorem}
    \begin{proof}
        \textsc{ExpSpace}-hardness of BC-indicator-SPP-U, and thus also of BC-indicator-SPP, follows from the hardness proof of BC-Parikh-SPP-U in Theorem~\ref{Parikh-SPP-complexity}. In that construction, the alphabet of protectable events consists of a single label corresponding to $t_{\mathit{new}}$, and the transition associated with this label can occur at most once in any run. Hence, under indicator semantics, the behavior of the system remains unchanged.
    
        Membership of both problems in \textsc{ExpSpace} follows from the same argument as in Theorem~\ref{Indikator-SPP-complexity}. Therefore, BC-indicator-SPP and BC-indicator-SPP-U are \textsc{ExpSpace}-complete.
    \end{proof}

\section{Discussion and Conclusion}
    From the constructions in the previous sections, it follows that it is not necessary to restrict the family $\chi$ to only Parikh functions or only indicator functions. Instead, the two types can be freely combined. In this section, let
    \begin{multline*}
        \Pi = \{\, \chi_a \mid a \in \Sigma_p,\;
        \chi_a \text{ is either the Parikh function} \\
        \text{or the indicator function} \}\,.
    \end{multline*}

    By combining the previous results, we obtain the following.
    \begin{corollary}~
        \begin{enumerate}
            \item For every instance of $\Pi$-SPP, there exists an instance of $\Pi$-SPP-U such that the costs of the respective $\Pi$-optimal protecting policies coincide.
        
            \item Both $\Pi$-SPP and $\Pi$-SPP-U for labeled Petri nets are solvable in exponential space.
        
            \item Both BC-$\Pi$-SPP and BC-$\Pi$-SPP-U for labeled Petri nets are \textsc{ExpSpace}-complete.
        \end{enumerate}
    \end{corollary}

    We also obtain the following stronger robustness result.
    \begin{theorem}
        BC-$\Pi$-SPP and BC-$\Pi$-SPP-U for labeled Petri nets are \textsc{ExpSpace}-complete even under any of the following restrictions:
        \begin{itemize}
            \item every transition has a unique label,
            \item the set of protectable events is a singleton,
            \item there is only one transition labeled by a protectable event, or
            \item the sole transition labeled by a protectable event can fire at most once.
        \end{itemize}
    \end{theorem}
    \begin{proof}
        The \textsc{ExpSpace}-hardness reduction in Theorem~\ref{Parikh-SPP-complexity} already satisfies each of the listed constraints.
    \end{proof}    

    To summarize, our results show that the synthesis of optimal protecting policies in labeled Petri nets is \textsc{ExpSpace}-complete, independently of whether events are counted by Parikh or indicator semantics, uniformly or not, or even under extreme structural simplifications. In the future, our goal is to characterize tractable fragmets and to investigate practical algorithms to solve the secret protection problem for labeled Petri nets.

\bibliographystyle{IEEEtranS}
\bibliography{mybib}

\begin{thebibliography}{1}
\providecommand{\url}[1]{#1}
\csname url@rmstyle\endcsname
\providecommand{\newblock}{\relax}
\providecommand{\bibinfo}[2]{#2}
\providecommand\BIBentrySTDinterwordspacing{\spaceskip=0pt\relax}
\providecommand\BIBentryALTinterwordstretchfactor{4}
\providecommand\BIBentryALTinterwordspacing{\spaceskip=\fontdimen2\font plus
\BIBentryALTinterwordstretchfactor\fontdimen3\font minus
  \fontdimen4\font\relax}
\providecommand\BIBforeignlanguage[2]{{%
\expandafter\ifx\csname l@#1\endcsname\relax
\typeout{** WARNING: IEEEtran.bst: No hyphenation pattern has been}%
\typeout{** loaded for the language `#1'. Using the pattern for}%
\typeout{** the default language instead.}%
\else
\language=\csname l@#1\endcsname
\fi
#2}}

\bibitem{ATIG2011}
M.~F. Atig and P.~Habermehl, ``On {Yen}'s path logic for {P}etri nets,''
  \emph{International Journal of Foundations of Computer Science}, vol.~22,
  no.~04, p. 783–799, 2011.

\bibitem{esparza}
J.~Esparza, ``Petri nets,'' 2018, lecture Notes.

\bibitem{MaC22}
Z.~Ma and K.~Cai, ``Optimal secret protections in discrete-event systems,''
  \emph{{IEEE} Transactions on Automatic Control}, vol.~67, no.~6, pp.
  2816--2828, 2022.

\bibitem{MaC25}
------, ``Secret protection in discrete-event systems with generalized
  confidentiality requirements,'' \emph{IEEE Transactions on Automatic
  Control}, vol.~70, no.~4, pp. 2321--2333, 2025.

\bibitem{MV2025}
T.~Masopust and J.~Ve{\v{c}}e{\v{r}}a, ``On the complexity of the secret
  protection problem for discrete-event systems,'' \emph{CoRR}, vol.
  arXiv:2509.14372, 2025, manuscript.

\bibitem{MatsuiC19}
S.~Matsui and K.~Cai, ``Secret securing with multiple protections and minimum
  costs,'' in \emph{58th {IEEE} Conference on Decision and Control}.\hskip 1em
  plus 0.5em minus 0.4em\relax {IEEE}, 2019, pp. 7635--7640.

\bibitem{Mur89}
\BIBentryALTinterwordspacing
T.~Murata, ``{Petri} nets: Properties, analysis and applications,''
  \emph{Proc.\ of the IEEE}, vol.~77, no.~4, pp. 541--580, 1989. [Online].
  Available: \url{http://dx.doi.org/10.1109/5.24143}
\BIBentrySTDinterwordspacing

\bibitem{yen1992unified}
H.-C. Yen, ``A unified approach for deciding the existence of certain {{P}etri}
  net paths,'' \emph{Information and Computation}, vol.~96, no.~1, pp.
  119--137, 1992.

\end{thebibliography}

\end{document}